\documentclass[a4document]{article}
\usepackage[margin=1in]{geometry}
\pdfoutput=1
\usepackage[utf8]{inputenc}
\usepackage[english]{babel}
\usepackage[T1]{fontenc}
\usepackage{amsthm}
\usepackage{amsmath,amssymb,dsfont}
\usepackage{graphicx}       
\usepackage{physics}
\usepackage{hyperref}
\usepackage{relsize}
\usepackage{tcolorbox}
\usepackage{thm-restate}
\hypersetup{
    colorlinks=true,
    linkcolor=blue,
    citecolor=magenta,
    filecolor=magenta,      
    urlcolor=blue,
    pdftitle={},
    pdfpagemode=FullScreen,
    }
\usepackage{soul}
\usepackage{mathtools}
\usepackage{enumerate}
\usepackage[normalem]{ulem}
\usepackage{authblk}

\usepackage[backend=biber,style=alphabetic,url=false,doi=false]{biblatex}
\newbibmacro{string+doi}[1]{
  \iffieldundef{doi}{
\iffieldundef{url}{#1}{\href{\thefield{url}}{#1}}
}{\href{http://dx.doi.org/\thefield{doi}}{#1}}}
\DeclareFieldFormat{title}{\usebibmacro{string+doi}{\mkbibemph{#1}}}
\DeclareFieldFormat[article]{title}{\usebibmacro{string+doi}{#1}}
\AtEveryBibitem{\clearfield{issn}\clearfield{primaryClass}\clearfield{note}
\clearfield{month}}
\usepackage{graphicx,color,colordvi}
\usepackage{tikz}
\usepackage{ifthen}
\usepackage{pgfplots}
\pgfplotsset{compat=1.9}
\usetikzlibrary{shapes,arrows.meta}
\usetikzlibrary{positioning}
\usetikzlibrary{shapes.geometric}

\usepackage{tikz}
\usepackage{quantikz}
\usetikzlibrary{shapes}
\usepackage[dvipsnames]{xcolor}
\usepackage{float}
\makeatletter
\newlength\min@xx

\makeatother

\newtheorem{theorem}{Theorem}

\newtheorem{lemma}{Lemma}
\newtheorem{corollary}{Corollary}
\newtheorem{definition}{Definition}
\newtheorem{remark}{Remark}

\DeclareMathOperator{\spec}{spec}
\DeclareMathOperator{\Prep}{Prep}

\newcommand{\cL}{\mathcal L}

\newcommand{\cO}{\mathcal O}
\newcommand{\cH}{\mathcal H}
\newcommand{\cB}{\mathcal B}

\newcommand{\kA}{\mathfrak{A}}

\DeclareMathOperator{\poly}{poly}
\DeclareMathOperator{\polylog}{polylog}
\newcommand{\eps}{\varepsilon}

\newcommand{\CC}{\mathbb{C}}
\newcommand{\RR}{\mathbb{R}}
\newcommand{\ZZ}{\mathbb{Z}}

\newcommand{\id}{\mathds{1}}

\definecolor{Turquoise}{HTML}{14C7DE}
\definecolor{SkyBlue}{HTML}{3498DB}
\definecolor{CoralDark}{HTML}{FF6F61} 
\definecolor{purplemod}{HTML}{AF82F3}
\definecolor{darkgreen}{RGB}{106, 134, 104}
\definecolor{darkgreen2}{RGB}{90, 168, 143}

\title{Dissipative microcanonical ensemble preparation from KMS-detailed balance}
\author[1]{Anirban N. Chowdhury}
\author[1,2]{Samuel O. Scalet}
\author[3]{Kunal Sharma}
\affil[1]{IBM Quantum, IBM T.J. Watson Research Center, Yorktown Heights, NY 10598, USA}
\affil[2]{UC Davis, CA 95616, USA}
\affil[3]{IBM Research, Chicago, IL 60606, USA}
\begin{document}

\date{}

\maketitle

\begin{abstract}

Stationary states of quantum many-body Hamiltonians are invariant under the Hamiltonian evolution. Besides ground and thermal states,  this class includes microcanonical ensembles that are of fundamental importance in statistical physics. We consider the preparation of general stationary states by leveraging recent advances in the field of open-system dynamics. In particular, constructions based on exact KMS-detailed balance with respect to Gibbs states of noncommuting Hamiltonians have only recently been proposed as a tool for their efficient preparation and, by extension to small temperatures, for ground state preparation. We extend these constructions to the problem of stationary state preparation, providing general criteria that characterize when such states have efficient implementations, along with specific results on the approximation of microcanonical ensembles. An interesting application of our work are tests of conjectured ensemble equivalences for local observables between microcanonical and Gibbs ensembles.
\end{abstract}
\section{Introduction}

Markov Chain Monte Carlo algorithms are a highly successful algorithmic primitive with vast applications.
In the regime of (classical) statistical physics their most famous implementation is the Metropolis-Hastings algorithm \cite{Hastings1970}, a physics inspired routine for simulating spin systems.
It is based on iterative updates of a spin configuration by randomly selecting a spin and flipping it according to some energy-dependent probabilistic rule, inspired by the physical, local process of thermalization.
It satisfies a so-called detailed balance condition, namely the transition rates $P_{ij}$ between configurations $i$ and $j$ of the stochastic process satisfy the equality
\[
P_{ij} p(j)=P_{ji}p(i)
\]
where $p$ is the probability distribution of the target state.

A quantum analogue of such a physical thermalization process has been known for many decades.
The Davies generator \cite{Davies1974} has been originally proposed as a model in mathematical physics for the process of thermalization of quantum systems weakly coupled to a large Markovian bath.
When restricting to commuting Hamiltonians, it turns out that the resulting evolution can be efficiently implemented due to its locality and has been proposed as an algorithmic tool \cite{kastoryano_quantum_2016,capel_modified_2020}.

An extension of such algorithmic ideas to the noncommuting setting, has only been achieved in recent years in a series of works \cite{chen2023quantumthermalstatepreparation,chen2023efficientexactnoncommutativequantum,Ding_2025,gilyen_quantum_2024,Temme_2011}.
The main challenge is that directly carrying over the formulation of Davies generators results in highly nonlocal dynamics and indeed naive attempts of quasi-local approximations do not succeed.
The aforementioned works overcame this issue, proposing the first quasi-local, approximately and exactly detailed balance Lindblad dynamics for the Gibbs state.
Common to all these approaches is their access model to the system Hamiltonian, which is only given via an implementation of its time evolution.

In this work we explore the applications of these recent algorithmic techniques to state preparation tasks beyond Gibbs states.
In the most general setting, we consider stationary states of the Hamiltonian $H$, i.e., states that can be expressed as functions of the Hamiltonian $\sigma=f(H)$ with $f$ some nonnegative real function.
We modify prior constructions to obtain \emph{KMS-detailed balance Lindbladians}.
\[
\cL(\rho)=-i[G,\rho]+\sum_{j\in\mathcal J}L_j\rho L_j^\dagger-\frac12\{L_j^\dagger L_j,\rho\}
\]
The KMS-detailed balance, analogously to the classical case, is a technical condition that ensures that the target state is a fixed point of the evolution $\sigma=e^{t\cL}(\sigma)$.
We provide efficient implementations of \emph{block-encodings} of the coherent and dissipative generators $G$ and $L_j$.
By referring to known Lindbladian simulation techniques these lead to efficient implementations of $e^{t\cL}$ \cite{cleve2019,li_simulating_2023}.
Our criteria are general and ensure efficiency for any target state defined in terms of a function $f$ that is \emph{twice real differentiable}.
As in the Gibbs sampling case, the algorithm accesses the Hamiltonian solely through an implementation of its time evolution.

For a concrete application, we apply the framework to the task of preparing a variant of the microcanonical ensemble. Recall that in statistical mechanics, the microcanonical ensemble is defined to be a uniform distribution over states with (or close to) a fixed value of the energy \cite{griffiths1965microcanonical,Dorje2007microcanonical}. Here, we consider the more general task of preparing a ``window state'', a uniform mixture of eigenstates of the Hamiltonian with eigenvalues within a given energy window. Such states have been considered in studying equivalence of microcanonical and canonical ensembles \cite{kuwahara2020gaussian}.
In addition, the approximations we use for window functions can yield ground states of gapped Hamiltonians.

The starting point of our construction is mostly the work of~\cite{Ding_2025}. 
On the analytical side, the authors give a complete characterization of KMS-detailed balance Lindbladians, which remains valid beyond Gibbs states.
On the implementation side, the discussion in~\cite{Ding_2025} is restricted to Gibbs samplers and we contribute alternative implementations extending to the more general class of states mentioned above.

It should be noted, that the algorithm presented here does not come with guarantees on the mixing time. The efficiency only concerns the Lindbladian evolution itself.
Fast or rapid mixing should in general not be expected for the entirety of states for which we can efficiently construct converging dynamics due to known hardness results as it includes Gibbs states as well as ground states with inverse polynomial gap.
We leave specific results for regimes where fast mixing can be proven as an open problem.
One apparent candidate regime is that of high-temperature Gibbs states for which prior works have shown fast mixing \cite{rouze2024efficientthermalizationuniversalquantum,rouze2024optimalquantumalgorithmgibbs}, but additional work is required to carry them over to the modified construction presented here.
Similarly, in analogy to prior works on the dissipative preparation of ground states \cite{zhan2025,tong2025,smid2025}, efficiency of our scheme in the corresponding setting can be expected.
On the more heuristic side, rather than working with such a strong theoretical result one could argue whether weaker notions of mixing could hold, such as ones restricted to local observables and/or restricted sets of input states. 
In particular, in light of ensemble equivalences, we expect favourable mixing times for microcanonical ensembles in energy ranges that correspond to the expected energy in high temperature Gibbs states.

\section{Preliminaries}
Throughout the paper we will consider a multipartite Hilbert space $\cH=\otimes_{v\in \Lambda}\cH_v$, where $\cH\cong\CC^d$ and $\Lambda$ is a finite set.
We further consider a Hamiltonian $H\in\cB(\cH)$.
We do not make further assumptions on $H$ but note that our algorithmic cost will be expressed in terms of Hamiltonian evolution time, a task that can be efficiently performed for instance for local Hamiltonians \cite{lloyd1996universal,berry2015simulating,low2019hamiltonian}.
We define an orthonormal basis of eigenstates $\ket{\psi_k}$ and the corresponding projectors $P_k=\ketbra{\psi_k}{\psi_k}$ for the associated eigenvalues $E_k$ such that $H=\sum_kE_kP_k$.
For a function $f:\spec(H)\mapsto\RR^+$, we define the full-rank state $\sigma=f(H)/\Tr[f(H)]$.
We will often consider continuous extensions $f:[-S,S]\mapsto\RR^+$ for some $S\ge\|H\|$ and $I=[-S,S]$.
A prominent example of such a state is the Gibbs state $\exp(-\beta H)/\Tr[\exp(-\beta H)]$. Here we will be interested in the window state (or microcanonical ensemble), which has the form
\[
\chi_{[a,b]}(H)/\Tr[\chi_{[a,b]}(H)]\,,
\]
where the indicator function $\chi_A(x)=1$ if $x\in A$ and $\chi_A(x)=\eta$ otherwise. The small constant $\eta>0$ ensures that the state is full-rank, which is needed for technical reasons. However, it can be chosen exponentially small in system size as we will discuss in a later section.

\paragraph{}
The paradigm of dissipative state preparation is to model time-evolution of a quantum system in contact with a Markovian heat-bath, so-called open system dynamics, which is designed such that the dynamics converges to the target state.
This is modeled by the Lindblad equation $\frac d{dt}\rho(t)=\cL(\rho(t))$
, i.e., $\rho(t)=e^{t\cL}(\rho(0))$.
For $e^{t\cL}$, to be a valid family of channels (completely positive, trace-preserving, linear maps) $\cL$ needs to be of the form
\[
\cL(\rho)=-i[G,\rho]+\sum_{j\in \mathcal J} L_j\rho L_j^\dagger-\frac12\{L_j^\dagger L_j,\rho\}\,,
\]
where $G=G^\dagger$ is a hermitian operator and $L_j$ are arbitrary operators forming the dissipative term.
Conversely every map of the above form generates a Lindbladian.
While the first term $G$ generates a coherent evolution, the $L_j$ form the dissipative part of the evolution.
Physically, $G$ is often identified with the system Hamiltonian and the $L_j$ are derived from a coupling limit, but for our purely algorithmic purposes we focus on designing and implementing these terms.
We call a Lindbladian faithful if it admits a full-rank invariant state $\cL(\sigma)=0$, i.e., $e^{t\cL}(\sigma)=\sigma$.
Let us further define the Hilbert-Schmidt inner product $\langle X,Y\rangle=\Tr[X^\dagger Y]$ on $\cB(\cH)$ and the  $\sigma$-KMS inner product
$\langle X,Y\rangle_{\sigma,KMS}=\Tr[X^\dagger\sigma^{1/2}Y\sigma^{1/2}]$.
A Lindbladian is called $\sigma$-KMS-detailed balance if $\cL^\dagger$ is self-adjoint with respect to this inner product.
This is a sufficient condition for stationarity of $\sigma$ by the standard calculation 
\[
0=\langle X,\cL^\dagger(\id)\rangle_{\sigma,KMS}=\langle\cL^\dagger(X),\id\rangle_{\sigma,KMS}=\langle\cL^\dagger(X),\sigma\rangle=\langle X,\cL(\sigma)\rangle
\]
for every $X\in\cB(\cH)$ and so $\cL(\sigma)=0$.
 In addition a faithful Lindbladian is called primitive if its fixed-point $\sigma$ is unique and thereby $\lim_{t\to\infty}e^{t\cL}(\rho)=\sigma$ $\forall\rho$.

\paragraph{}
For implementation purposes we will make use of efficient block-encodings.
\begin{definition}[Approximate block-encoding]
We say that a unitary $U$ is a $(\alpha,m,\eps)$-block encoding of a matrix $B$ if
\[
\left\|B/ \alpha-(\bra{0^{\otimes m}}\otimes\id_\cH) U(\ket{0^{\otimes m}}\otimes\id_\cH)\right\|\le\eps
\]
\end{definition}
\cite{li_simulating_2023} provides an efficient simulation algorithm of Lindbladian evolutions given an approximate block encoding of its jump operators, which is why we focus in the implementation of such a block encoding.
The block encoding will be based on a block encoding of jump proposals $\{A_a:a\in\mathfrak{A}\}$  of the following form:\footnote{While \cite{li_simulating_2023} names individual block-encodings for each $L_a$ as a requirement, in fact the implementation only makes use of this joint block-encoding.}
\[
U_\mathfrak A=\sum_{a\in\mathfrak A}\ketbra{a}{a}\otimes A_a
\]
Note that for the standard choice $\{A_a:a\in\mathfrak{A}\}=\{X_i,Y_i,Z_i:i\in\Lambda\}$, where the index register is of size $\lceil\log_2(3|\Lambda|)\rceil$, this unitary can be implemented in circuit size $\mathcal{O}(|\Lambda|\log(|\Lambda|))$ without ancillas using a sequence of multi-qubit Toffoli gates~\cite{nie2024quantumcircuitmultiqubittoffoli}.
While we do not want to explicitly assume the above choice of jump proposals,\footnote{This is because "block updates" with multi-spin jump proposals have proven to be useful in some settings to ensure good mixing times \cite{ding2025polynomialtimepreparationlowtemperaturegibbs,guo2018}.} we will henceforth assume $\|A_a\|\le1$ to simplify notation. This comes with no loss of generality as it can always be achieved by rescaling the Lindbladian and thereby the mixing time.

\paragraph{}
We define the Fourier series for a continuous function $f:[-S,S]^D\to\CC$ such that $f(S)=f(-S)$ with the convention
\[
f(x_1,\ldots,x_D)=\sum_{(x_1,\ldots,x_D)\in\ZZ^D}f_{n_1,\ldots,n_D} e^{i\langle x,n\rangle\tau},
\]
where $\tau=\pi/S$, i.e.,
\[
f_{n_1,\ldots,n_D}=\frac{1}{(2S)^D}\int_{[-S,S]^D} e^{-i\langle n,x\rangle\tau} f(x_1,\ldots,x_D)d(x_1,\ldots,x_D)
\]

For a finite domain $A\subset\RR^D$, with $|A|=\int_Ad\mu$, we define the normalized $L_1$-norm
\[
\|f\|_{L_1}=\frac1{|A|}\int_A|f|\,d\mu\,.
\]
\section{Construction of algorithm and general analysis}
\subsection{KMS-symmetric jump operators}
We start from the following known characterization of KMS-detailed balance Lindbladians, see \cite{FAGNOLA_2007} for an earlier work giving an equivalent characterization.
\begin{theorem}[{\cite[Theorem 10]{Ding_2025}}]\label{thm:KMSDBCharacter}
Let $\sigma>0$ be a state. The $\sigma$-KMS-DB Lindbladians with respect to $\sigma$ are given uniquely as
\[
\cL^\dagger(X)=i[G,X]+\sum_{j
\in J}\left(L^\dagger_j XL_j-\frac12\{L_j^\dagger L_j,X\}\right)\,,
\]
where
\begin{equation}\label{eq:KMScondition}
\Delta_\sigma^{-1/2}L_j=\sigma^{-1/2} L_j \sigma^{1/2}=L_j^\dagger
\end{equation}
and
\[
G=i\tanh(\log(\Delta_\sigma^{1/4}))\left(-\frac12\sum_aL_a^\dagger L_a\right)\,.
\]
\end{theorem}

\subsection{Energy-domain representation}
While the above fully characterizes KMS-detailed balance Lindbladians, it is crucial to consider efficient implementations.
To that end, we rewrite the Lindblad operators as a product of an operator with an efficient time-domain implementation and a self-adjoint "jump-proposal".\footnote{This is in analogy with classical MCMC methods, where a self-adjoint spin-flip is proposed and its acceptance probability with energy-dependent weight ensures detailed balance.}
Following the approach of~\cite{Ding_2025}, we define the Lindbladian terms using the following simple observation:
For any self-adjoint operator $A$, the operator $L=\sigma^{1/4}A\sigma^{-1/4}$ fulfills the condition Eq.~\eqref{eq:KMScondition}.

\begin{definition}\label{def:jumps}
We define the jump operators
\begin{equation}\label{eq:defLa}
L_a=\sum_{k,l}\hat g(E_k,E_l)P_kA_aP_l\,,
\end{equation}
with
\[
\hat g(E_1,E_2) =\sqrt[4]{\frac{f(E_1)}{f(E_2)}}\kappa(E_1,E_2)
\]
where $A_a$ are self-adjoint jump proposals taken from a set indexed by $a\in\mathfrak{A}$ and $\kappa: [-S,S]^2\to\CC$ satisfies $\kappa(E_k,E_l)=\overline{\kappa(E_l,E_k)}$.

For the coherent term we write in an energy decomposition
\[
G=i \sum_{k,l}\tanh(\log(\frac{f(E_k)}{f(E_l)})/4)P_k\left(-\frac12\sum_aL_a^\dagger L_a\right)P_l.
\]
\end{definition}
\begin{remark}
The main difference to previous construction is that the function $\hat g(E_1,E_2)$ depends on $E_1$, $E_2$ individually rather than only on their difference $E_1-E_2$, an implementation approach noted in \cite{chen2023quantumthermalstatepreparation}.
In fact, if $\hat g(E_1,E_2)=h(E_1-E_2)$, detailed balance respect to a fixed point $f(H)$ implies that $f$ is an exponential function under some mild assumptions.

Specifically, this is the case if detailed balance holds for the transition part (common to all the constructions in the literature), the construction works for arbitrary $H$ (with arbitrary eigenvalues) and $A_a$ and assuming that $f$ is continuous. To see this, observe that the detailed balance condition applied to observables $P_k$, $P_l$, and the jump proposal $A_a=\sum_{k,l}\ketbra{\psi_k}{\psi_l}$ 
\[
|h(E_k-E_l)|^2 f(E_l)=\langle L_a^\dagger P_k L_a,P_l\rangle_{\sigma,KMS}=\langle P_k,L_a^\dagger P_l L_a \rangle_{\sigma,KMS}=|h(E_l-E_k)|^2 f(E_k),
\]
which reduces to the functional equation of the exponential.
\end{remark}

In the above definition of $L_a$, the $\kappa(E_k,E_l)$ is a filter function, while the ratio of $f$ at different energies ensures detailed balance. In fact, the filter function corresponds to replacing the self-adjoint jump proposal $A_a$ by another self-adjoint operator:
\[
L_a=\sum_{k,l}\sqrt[4]{\frac{f(E_k)}{f(E_l)}}P_k \tilde A_aP_l\,,
\]
where
\[
\tilde A_a=\sum_{k,l}\kappa(E_k,E_l) P_k A_aP_l=\left(\sum_{k,l}\overline{\kappa(E_k,E_l)} P_lA_aP_k\right)^\dagger=\left(\sum_{k,l}\kappa(E_k,E_l)P_lA_aP_l\right)^\dagger=\tilde A_a^\dagger\,.
\]
Therefore, the jump operators $L_a$ fulfill the conditions of Theorem~\ref{thm:KMSDBCharacter} for the state $\sigma=f(H)/\Tr[f(H)]$.

\subsection{Choice of filter function}
We would like to implement the Lindbladian above efficiently, i.e., a unit-time evolution with the Lindbladian should be possible with a polynomial-sized quantum circuit. Key to this is choosing the filter function in such a way that the $G$ and $L_a$ terms can be efficiently block-encoded. This turns out to be non-trivial, as a poor choice of $\kappa$ can immediately lead to a very large mixing time. Below we motivate our choice of filter functions and discuss how this affects the running-time.

To achieve the efficiency, the conditions on $\kappa$ or rather more directly on $\hat g$ are explained in the following.
Simply speaking, differentiability results in a faster runtime. In particular, a polynomial runtime can be ensured if $\hat g$ is twice differentiable and uniformly bounded.
In the rest of the paper we assume the following filter function ensuring these conditions, see Lemma~\ref{lem:filter}. 
Let $f(x)=\exp(\Phi(x))$. Let $\Phi\le0$ be $L$-Lipschitz. Let
\begin{align*}
\kappa(E_k,E_l)&=\nu_{L^2S^2/4,S}(E_k-E_l)\\
\nu_{C,\zeta}&=\exp(-\sqrt{1+C(1-\cos(x\pi/\zeta))})\,.
\end{align*}

The consideration of good mixing times is more subtle. We do not address bounds on the mixing time in this manuscript, but remark the following regarding the motivation for the choice of filter function.
For constant $f$ without a filter function and for $A_a$ Pauli jumps, the Lindbladian reduces to a product of depolarizing channels and mixes rapidly.
$L_a$ is defined as an elementwise product of $\hat g$ and $A_a$, see Eq. \eqref{eq:defLa}. Let us assume the self-adjoint jump-proposal $A_a$ is a local operator.
In the case of local Hamiltonians, the matrix elements of $A_a$ are concentrated around the diagonal due to the following bound~\cite{Arad_2016}
\[
\|\Pi_{[\eps',\infty]}A\Pi_{[-\infty,\eps]}\|\le \|A\|\exp(\lambda(\eps'-\eps)-2R)
\]
where $\Pi$ are projectors onto the energy eigenspaces within the respective intervals, and $\lambda$ and $R$ are universal constants only dependent on the locality and strength of the interactions and the size of the support of $A$ but not on the system size or $\eps$, $\eps'$.

This means that elements far from the diagonal (large $|E_k-E_l|$), which are most suppressed by the filter function $\kappa(E_k,E_l)$, the elements of $A_a$ were small in the first place.
Close to the diagonal on the other hand, for controlled derivatives of $f$, the matrix elements of $\kappa$ between nearby energies remain far from zero ensuring good connectivity, by not suppressing the large elements of $A_a$.

On the other hand $\hat g$ needs to be uniformly bounded to ensure efficiency. This is because its $L_1$ norm contributes to the implementation cost. This is ensured due to the decay of the filter function based on the Lipschitz constant of $\Phi$, see Lemma~\ref{lem:filter}.

\subsection{Time-domain} Note that the above decompositions are not aligned with their implementations with sums ranging over individual eigenvalues and projectors with perfect energy-resolution.
The solution to this, as it has been proposed in \cite{chen2023quantumthermalstatepreparation,chen2023efficientexactnoncommutativequantum,Ding_2025}, is to consider the jump operators in the time-domain. We rewrite the sum in terms of its Fourier decomposition which takes the form of a sum over discretized time steps of time-evolved jump proposals. By truncating these sums and thereby bounding the overall Hamiltonian evolution time, we obtain efficient implementations (block-encodings) of $L_a$ and $G$:

Since we need to only obtain the correct values of $\hat g$ and $\hat w$ for $E\in[-\|H\|,\|H\|]\subset [-S,S]:=I$, where $S\ge\|H\|$ is a constant, we can actually use Fourier series instead of transforms.
This deviates from the formalism in prior works. \cite{chen2023quantumthermalstatepreparation,chen2023efficientexactnoncommutativequantum} instead bound the discretization error of the Fouriertransform, and in \cite{Ding_2025} the vanishing discretization error is indirectly recovered from the Plancherel formula by introducing compactly supported filter functions.
The  assumption of a bounded interval as the domain for $f$ contains these technicalities in the choice of state function hence simplifying the implementation.

Representing the function $\hat g$ via its Fourier series we obtain
\begin{align*}
L_a&:=\sum_{k} \hat g(E_k,E_l)P_k A_a P_l\\
&=\sum_{(n_1,n_2)\in\mathbb Z^2}g_{n_1,n_2}e^{in_1\tau H}A_ae^{in_2\tau H}
\end{align*}
where we determine the coefficients from the Fourier series formula
\[
g_{n_1,n_2}=\frac{1}{2S}\int_{I} e^{-i(n_1 E_1+n_2 E_2)\tau}\hat g(E_1,E_2)\,d(E_1,E_2)\,.
\]
And similarly for the coherent term with $\hat w(E_1,E_2)=i\tanh(\log(f(E_k)/f(E_l))/4)$
\begin{align*}
G&=\sum_{k,l}i\tanh(\log(\frac{f(E_k)}{f(E_l)})/4)P_k V P_l\\
&=\sum_{n_1,n_2}w_{n_1,n_2} e^{i\tau Hn_1}Ve^{i\tau Hn_2}\,,
\end{align*}
where $V=-\frac12\sum_aL_a^\dagger L_a$ and the equivalent
\[
w_{n_1,n_2}=\frac{1}{4S^2}\int_{I} e^{-i(n_1E_1+n_2E_2)\tau}\hat w(E_1,E_2)\,d(E_1,E_2)\,.
\]
We note that these representations can be implemented using Hamiltonian simulation techniques after truncation to finite sums.
The error of such truncations is what we discuss next.

\subsection{Truncation and implementation}
In order to obtain decay results on the Fourier coefficients, we impose certain assumptions on the function $f$.
In particular, bounds on the derivatives of $f$ will dictate the complexity of the algorithm as the degree of differentiability implies the rate of decay of Fourier coefficients.
For functions that do not satisfy these assumption in the first place such as the exact window functions defining microcanonical ensembles we need to choose suitable differentiable approximations as we will discuss in the following section. 
However, we deviate from the approach in ~\cite{Ding_2025} in that we work with $k$-times differentiable rather than Gevrey-class functions.
While it results in a slightly suboptimal dependence of our scheme in the implementation error, we prefer this formulation since it relaxes the assumptions on the defining functions $f$, resulting in more straightforwardly checkable conditions.

Let $k\ge2$. We define the state weights as
\[
f(E)=e^{\Phi(E)}
\]
where $\Phi:\RR\mapsto\RR^-_0$ is $k$-times differentiable, periodic and $L$-Lipschitz.

Note that modifying the function $\Phi$ outside the interval $[-\|H\|,\|H\|]$ does not change the target distribution, which provides additional freedom that can be used to ensure the periodicity.
For instance, in the case of thermal states, we define the $2S$-periodic function $\Phi(E)=-\beta (E)$ for $E\in[-\|H\|,\|H\|]$ and otherwise define the function via a k-times differentiable polynomial interpolation on $[\|H\|,2S-\|H\|]$ matching the derivatives at the boundaries.                                    

We aim to implement
\[
\overline{L}_a:=\sum_{n_1,n_2\in\mathbb Z,|n_1|,|n_2|\le M}g_{n_1,n_2} e^{i\tau Hn_1}A_ae^{i\tau Hn_2}
\]
which can be done given conditional Hamiltonian evolution
\[
U_H=\sum_{m=-M}^M \ketbra{m}{m}\otimes e^{im\tau H}
\]
for evolution times up to $M\tau$.\footnote{We assume a perfect implementation of the Hamiltonian evolution, but bound the evolution time by a polynomial. Replacing these evolutions with their implementations only incurs a polynomial overhead in the system size and interaction strength and sublogarithmic in the desired error~\cite{Low_2017}.}

Using Lemma~\ref{lem:fourierTailBounds}, which states standard tail bounds on Fourier coefficients of $k$-times differentiable functions, we bound
\begin{align*}
\|L_a-\overline{L}_a\|&\le\sum_{(n_1,n_2)\in\ZZ^2\setminus[-M,M]^2}|g_{n_1,n_2}|\\
&\le \frac{S^k}{\pi^k}\left(\frac{8S^k}{\pi^k}\left\|\partial^{(k,k)}\hat g\right\|_{L^1}+\left\|\partial^k_1\hat g\right\|_{L^1}+\left\|\partial^k_2\hat g\right\|_{L^1}\right)\frac{2M^{1-k}}{k-1}
\end{align*}

For $G$ we define
\[
\overline{G}=-\frac12\sum_{(n_1,n_2)\in\ZZ^2\setminus[-M',M']^2}w_{n_1,n_2}e^{i\tau H n_1}\sum_{a\in\mathfrak{A}} \overline{L}_a^\dagger \overline L_ae^{i\tau H n_2}
\]
then, again applying Lemma~\ref{lem:fourierTailBounds}
\begin{align}\label{eq:Gerr}
\left\|G-\overline{G}\right\|&\le \left\|G-\sum_{(n_1,n_2)\in[-M',M']^2}w_{n_1,n_2}e^{i\tau H n_1}Ve^{i\tau H n_2}\right\|\nonumber\\
&\phantom{==}+\sum_{n_1,n_2\in[-M',M']^2}|w_{n_1,n_2}|\left\|\frac12\sum_{a\in A}\left(L_a^\dagger L_a-\overline{L}_a^\dagger \overline{L}_a\right)\right\|\nonumber\\
&\le\|V\|\left(\frac{8S^k}{\pi^k}\left\|\partial^{(k,k)} \hat w\right\|_{L^1}+\left\|\partial_1^k\hat w\right\|_{L^1}+\left\|\partial_2^k\hat w\right\|_{L^1}\right)\frac{2S^kM'^{1-k}}{\pi^k(k-1)}\nonumber\\
&\phantom{==}+\left(9\|\hat w\|_{L_1}+\left(\frac{8S^k}{\pi^k}\left\|\partial^{(k,k)} \hat w\right\|_{L^1}+\left\|\partial_1^k\hat w\right\|_{L^1}+\left\|\partial_2^k\hat w\right\|_{L^1}\right)\frac{2S^k}{\pi^k(k-1)}\right)\nonumber\\
&\phantom{====}\times|\mathfrak A|\frac12\left(\left\|L_a\right\|+\left\|\overline{L_a}\right\|\right)\left\|L_a-\overline{L_a}\right\|\,,
\end{align}
where
\begin{align*}
    \|L_a\|&=\left\|\sum_{(n_1,n_2)\in\mathbb Z^2}g_{n_1,n_2}e^{in_1\tau H}A_ae^{in_2\tau H}\right\|\le\|g\|_{l^1}\\
    \|V\|&\le\frac{|\mathfrak A|}2 \max_a\|L_a\|^2\le\frac{|\mathfrak A|\|g\|_{l_1}^2}{2}\\
    \|\hat w\|_{L_1}&\le \sup_{E_1,E_2\in[-S,S]} |\hat w(E_1,E_2)|\le1\,.
\end{align*}
Then, assuming the truncation is chosen such that $\|L_a-\overline{L}_a\|\le1$
\begin{align*}
    \|\overline{L}_a\|\le\|L_a\|+\|L_a-\overline{L}_a\|\le 1+\|g\|_{l^1}
\end{align*}

Now let us consider the implementations and scalings of a block encoding of $L_a$ and $G$ starting with $G$.

We now require that the error bound in Equation~\eqref{eq:Gerr} is bounded by $\eps$. To bound the first term by $\eps/2$ it suffices to choose an $M'$ with
\[
M'=\cO\left(\sqrt[k-1]{\frac{|\mathfrak A|\|g\|_{l_1}^2\left(S^k\left\|\partial^{(k,k)} \hat w\right\|_{L^1}+\left\|\partial_1^k\hat w\right\|_{L^1}+\left\|\partial_2^k\hat w\right\|_{L^1}\right)S^k}{\eps}}\right)
\]
Bounding the second term by $\eps/2$ is in fact fulfilled if
\begin{equation}\label{eq:GrequiredBoundL}
\|L_a-\overline{L}_a\|\le o\left(\frac{\eps}{|\mathfrak A|\left(1+\left(S^k\|\partial^{(k,k)}\hat w\|_{L_1}+\|\partial^k_1\hat w\|_{L_1}+\|\partial_2^k \hat w\|_{L_1}\right)S^k\|g\|_{l^1}\right)}\right),
\end{equation}
which we ensure by choosing a sufficiently large $M$ of
\begin{align*}
M&=\mathcal O\Big(\sqrt[k-1]{1/\eps}\sqrt[k-1]{|\mathfrak A|\left(1+\left(S^k\|\partial^{(k,k)}\hat w\|_{L_1}+\|\partial^k_1\hat w\|_{L_1}+\|\partial_2^k \hat w\|_{L_1}\right)S^k\|g\|_{l^1}\right)} \\
&\quad\times\sqrt[k-1]{\left(S^k\|\partial^{(k,k)}\hat g\|_{L_1}+\|\partial^k_1\hat g\|_{L_1}+\|\partial_2^k \hat g\|_{L_1}\right)S^k}\Big)
\end{align*}
This is already stronger than our requirement on the implementation error of $L_a$ directly.

For the block encoding we need gates that encode the Fourier coefficients, i.e., they map $\ket0$ to
\[
\Prep_w\ket0=\frac1{\sqrt{Z_w}}\sum_{(n_1,n_2)\in[-M',M']^2}\sqrt{w_{n_1,n_2}}\ket{n_1,n_2}
\]
for $G$ with (see Lemma~\ref{lem:fourierL1Bounds})
\[
Z_w=\|w\|_{l^1([-M,M]^2)}\le20\left(1+\left(\frac{8S^k}{\pi^k}\left\|\partial^{(k,k)}\hat w\right\|_{L^1}+\left\|\partial^k_1\hat w\right\|_{L^1}+\left\|\partial^k_2\hat w\right\|_{L^1}\right)\frac{2S^k}{\pi^k}\right)^{\frac2{k+1}}
\]
and for $L_a$ we need 
\[
\Prep_g\ket{0}=\frac{1}{\sqrt{Z_g}}\sum_{(n_1,n_2)\in[-M,M]^2}\sqrt{g_{n_1,n_2}} \ket{n_1,n_2}\,.
\]
with
\[
Z_g=\|g\|_{l^1([-M,M]^2)}\le \|g\|_{l^1}\,.
\]
This $l^1$-norm used here and above can be bounded as (see Lemma~\ref{lem:fourierL1Bounds})
\[
\|g\|_{l^1}\le 10(1+\|\hat g\|_{L^1})\left(1+\left(\frac{8S^k}{\pi^k}\left\|\partial^{(k,k)}\hat g\right\|_{L^1}+\left\|\partial_1^k\hat g\right\|_{L^1}+\left\|\partial_2^k\hat g\right\|_{L^1}\right)\frac{S^k}{\pi^k}\right)^{\frac2{k+1}}\,.
\]

We now have all the ingredients and error bounds to construct a block encoding of $G$ and $L_a$ respectively.

The components are for $L_a$, see Fig~\ref{fig:circ_jump},
\begin{itemize}
    \item $\Prep_g$: For classically precomputed coefficients this can in general be implemented in polynomial time in $M$ (as any unitary on $\log(M)$ qubits). This already results in polynomial-time guarantees for our results, but prevents us from giving explicit degrees. However, we typically expect that a coherent implementation of numerical integration algorithms results in a $\polylog(M)$ runtime for this part. This is the case if the function is piecewise analytic (as all functions considered henceforth) and under some mild assumptions on the locations of their complex zeros \cite{trefethen2008}. We omit a detailed discussion of such bounds since they are purely classical numerical analysis problems and denote the maximum cost of the prep gate of $g$ and $w$ by $CPrep(M)$ in the following.

    \item Controlled Hamiltonian evolution for time up to $T=M\tau$, a well-studied problem, see for example \cite{Low_2017}.
    \item The block-encoding  $U_{\mathfrak{A}}$.
    \item $\Prep_g^\dagger$.
\end{itemize}

\begin{figure}[H]
    \centering
    \begin{quantikz}[wire types={b,b,b,b},thin lines]
\lstick{$\ket{0}$}  & \gate[2]{ \text{Prep}_g } &   &&\ctrl{3}&\gate[2]{\Prep_g^\dagger}&\\
\lstick{$\ket{0}$}  &  & \ctrl{2}  &&&&\\
\lstick{$\ket{a}$}  &              &                  & \gate[2]{U_{\mathfrak{A}}} &&&\\
\lstick{$\rho$}     &                  & \gate{e^{iH\tau n_1}}          &&\gate{e^{iH\tau n_2}}& &
\end{quantikz}
    \caption{Circuit for block encoding of the jump operators. Not depicted are additional ancillas that may but need not be necessary for implementing $U_{\kA}$. It might seem unconventional that both Hamiltonian evolutions come with a positive sign. This is purely a convention of the Fourier transform, but we want to emphasize that the two evolution times are independent and do \emph{not} form a Heisenberg time-evolution of the operator as in the special case of the Gibbs samplers.}
    \label{fig:circ_jump}
\end{figure}
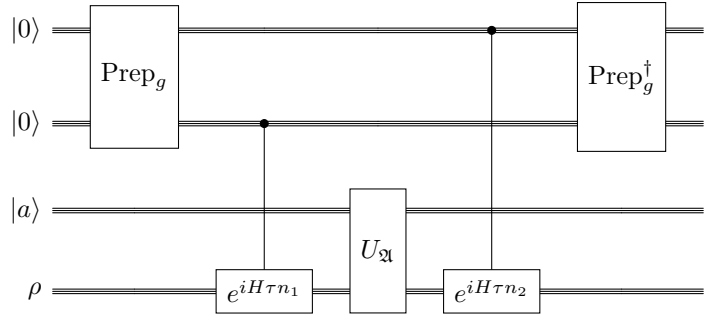

For $G$ the components are similar: $\Prep_w$, conditional Hamiltonian evolution up to time $T'=M'\tau$, two of the above encodings of $L_a$ sharing \textit{the same} index register for $A_a$, another Hamiltonian evolution, and $\Prep_w^\dagger$, see Fig~\ref{fig:coherent}.
We obtain $(\alpha_{L/G},a_{L/G},\eps)$-block encodings of $L_a$ and $G$ with scaling factor
\[
\alpha_L=Z_g,\quad\alpha_G=Z_g^2Z_w\,,
\]
and ancilla register size (including the size of the index register for $\mathfrak{A}$)
\[
a_L=2\log(M)+\log|\mathfrak{A}|,\quad a_G=\log(M')+2\log(M)+\log(|\mathfrak{A}|)
\]
with the asymptotic bounds on $M, M'$ given above, assuming an ancilla free block encoding of $U_\mathfrak{A}$, and $\eps$ which can be freely chosen.
\begin{figure}[h!]
    \centering
    \begin{quantikz}[wire types={b,b,b,b,b,b},thin lines]
\lstick{$\ket{0}$}  & \gate[2]{ \text{Prep}_w } & \ctrl{4}  &&&&\gate[2]{\Prep_w^\dagger}&\\
\lstick{$\ket{0}$}  &  &   &&&\ctrl{3}&&\\
\lstick{$\ket{0}$}  &  & & \gate[3]{U_{L_a}}&&&&\\
\lstick{$\ket{a}$}  &              &                  & &\gate[3]{U_{L_a^\dagger}}&&&\\
\lstick{$\rho$}     &                  & \gate{e^{iHt}}          & &&\gate{e^{iHt}}&&\\
\lstick{$\ket{0}$}  &  &   &&&&&
\end{quantikz}
    \caption{Circuit for block encoding of the coherent term. The gates $U_{L_a}$ are shorthand for the entire circuit in Figure~\ref{fig:circ_jump}. }
    \label{fig:coherent}
\end{figure}
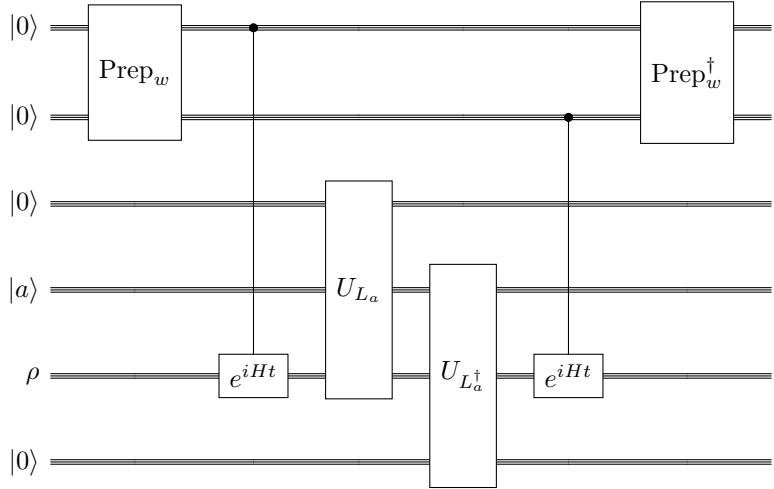
\begin{remark}
    One may note the similarity of the norms in the above analysis to Sobolev norms, which are sums over $L_p$-norms with fixed $p$ of derivatives of a function up to a finite order. Nevertheless, we refrain from introducing the concept since Sobolev spaces consider weak derivatives, which would introduce additional difficulties in working with a \emph{point spectrum} as the one of $H$ in finite-dimensional systems. Under the assumption of a pointwise $k$-times differentiable function, however, the norm bounds would be subsumed by $W^{k,1}$ Sobolev norm bounds.
\end{remark}

\subsection{Complexity via block encoding}
With the block encodings from the previous sections, in light of the following theorem, an efficient implementation of the detailed balance Lindbladian becomes an immediate consequence.
\begin{theorem}{\cite[Theorem 11]{li_simulating_2023}}\label{thm:liSim}
Suppose we are given an $(\alpha_0,a,\eps')$-block-encoding $U_H$ of $H$ and an $(\alpha_L,a,\eps')$-block-encoding $U_{L_j}$ for each $L_j$. Let $\|\cL\|_{be}=\alpha_0+|\mathcal A|\alpha_L^2/2$. For all $t,\eps\ge0$ with $\eps'\le\eps/(t\|\cL\|_{be})$, there exists a quantum algorithm for simulating $e^{\cL t}$  to error $\eps$ in diamond norm using
\[
\cO\left(t\|\cL\|_{be}\frac{\log(t\|\cL\|_{be}/\eps)}{\log(\log(t\|\cL\|_{be}/\eps))}\right)
\]
queries to $U_H$ and $U_{L_j}$ and
\[
\cO\left(t|\mathfrak{A}|\|\cL\|_{be}\left(\frac{\log(t\|\cL\|_{be}/\eps)}{\log(\log(t\|\cL\|_{be}/\eps))}\right)^2\right)
\]
additional 1- and 2-qubit gates.
\end{theorem}
Without restating the tighter bounds in the previous section, we summarize the runtime in the following Lemma:
\begin{corollary}
For the state $\sigma=f(H)>0$, there exists a Lindbladian $\cL$ that satisfies $\sigma$-KMS detailed balance. There exists a quantum channel $\Phi$ with $\|\Phi-e^{t\cL}\|_\diamond\le\eps$ that can be implemented in circuit complexity
\begin{align*}
&t|\kA|^2\times\poly\log(t/\eps)\,\times \\
& \sqrt[k+1]{\left(\|\partial^{(k,k)}\hat w\|_{L_1}+\|\partial^k_1\hat w\|_{L_1}+\|\partial_2^k \hat w\|_{L_1}\right)^3
\left((1+\|\partial^{(k,k)}\hat g\|_{L_1}+\|\partial^k_1\hat g\|_{L_1}+\|\partial_2^k \hat g\|_{L_1})(1+\|\hat g\|_{L^1})\right)^{5}S^{7k}}
\end{align*}
times the cost of the gates $Prep_{g/w}$ plus the cost of Hamiltonian simulation for time
\begin{align*}
S\sqrt[k+1]{\frac{|\kA|S^{5k}}{\eps}\left(1+\|\partial^{(k,k)}\hat w\|_{L_1}+\|\partial^k_1\hat w\|_{L_1}+\|\partial_2^k \hat w\|_{L_1}\right)\left(1+\|\partial^{(k,k)}\hat g\|_{L_1}+\|\partial^k_1\hat g\|_{L_1}+\|\partial_2^k \hat g\|_{L_1}\right)^{\frac{k+5}{k+1}}(1+\|\hat g\|_{L^1})^2}
\end{align*}
\end{corollary}
The degrees are not optimal (among others they incorporate poly-log dependencies). Also note that in the case of Gevrey functions $w$ and $g$, we expect that the dependence on $\eps$ becomes poly-logarithmic instead of $\sqrt[k+1]\cdot$, but the above bounds are convenient when only differentiability is known.

\section{Applications}
\subsection{Microcanonical ensembles}

We can apply the above bounds by defining a function $f=e^\Phi$ and thereby a fixed-point.
Our goal is to approximate a window function $\chi_{[b,c]}$.
Note that we do not know the normalization constant that defines the trace-normalized fixed-point, but since the space of fixed points is linear we can just work with the unnormalized version.
Therefore, we will choose $f$ such that $\|f\|_{L^\infty}=1$.

Since, the sharp window function is not differentiable, we use a $k$-times continuously differentiable interpolation $f$.
Let $\gamma:[0,1]\to[0,1]$ be a k-times continuously differentiable function such that
\begin{align*}
    \gamma(0)&=0\\
    \gamma(1)&=1\\
    \gamma^{(l)}(0)&=\gamma^{(l)}(1)=0\qquad\forall 1\le l\le k\,,
\end{align*}
which can be easily found using a polynomial ansatz and solving a linear system of equations. We define $\Phi:[-S,S]\to[-\log(\eta),0]$
\begin{align*}
    \Phi:x\mapsto \begin{cases}
        -\log(\eta)&\textrm{if}\quad x\le b-\delta \\
        -\log(\eta)\gamma((b-x)/\delta)&\textrm{if}\quad b-\delta<x\le b\\
        0&\textrm{if}\quad b<x\le c\\
        -\log(\eta)\gamma((x-c)/\delta)&\textrm{if}\quad c<x\le c+\delta\\
        -\log(\eta)&\textrm{if}\quad c+\delta<x
    \end{cases}
\end{align*}
Here we assume that $S\ge\delta+\max\{|b|,|c|\}$.
Note that $\left\|\Phi^{(l)}\right\|_{L^\infty}=\frac{-\log(\eta)}{\delta^l}\left\|\gamma^{(l)}\right\|_{L^\infty}$, where $\left\|\gamma^{(l)}\right\|_{L^\infty}$ is a universal constant (given via the solution of the polynomial interpolation). The Lipschitz constant $L$ is equal to $\|\Phi^{(1)}\|_{L^\infty}$.

For $\hat g(E_k,E_l)=e^{(\Phi(E_k)-\Phi(E_l))/4}\kappa(E_k,E_l)$, note that by the chain rule
\[
\frac{d^m}{dE_k^m}\frac{d^n}{dE_l^n}\hat g(E_k,E_l)=\hat g(E_k,E_l)\poly\left(\Phi'(E_k),\Phi'(E_l),\ldots,\Phi^{(n)}(E_k),\Phi^{(m)}(E_l),1,\ldots,\frac{1}{\kappa}\frac{d^m}{dE_k^m}\frac{d^n}{dE_l^n}\kappa(E_k,E_l)\right)
\]
where for each $m,n$, for every monomial in the polynomial the sum over orders of derivatives times their respective degree is less or equal to $m+n$ and thereby the polynomial scales as $\cO((-\log(\eta)/\delta)^{m+n})$, while $\hat g(E_k,E_l)\le1$.
Therefore, we obtain the bound
\begin{align*}
\left\|\partial_{1,2}^{(l)}\hat g\right\|_{L^1}\le\|\partial^{(l)}\hat g\|_{L^\infty}&=\mathcal O\left(\left(\frac{-\log(\eta)}{\delta}\right)^l\right)\\
\left\|\partial^{(l,l)}\hat g\right\|_{L^1}&=O\left(\left(\frac{-\log(\eta)}{\delta}\right)^{2l}\right)
\end{align*}

For
\begin{equation}\label{eq:hatwWindow}
\hat w(E_1,E_2)=i\tanh(\frac{\Phi(E_1)-\Phi(E_2)}{4})\,,
\end{equation}
the argument is similar with
\begin{align*}
\partial^{(l_1,l_2)}&\left(\tanh(\frac{\Phi(E_1)-\Phi(E_2)}{4})\right)=\poly\Big(\tanh'\left(\frac{\Phi(E_1)-\Phi(E_2)}4\right),\ldots,\\
&\tanh^{l_1+l_2}\left(\frac{\Phi(E_1)-\Phi(E_2)}4\right),\Phi'(E_1),\ldots,\Phi^{(l_1)}(E_1),\Phi'(E_2),\ldots,\Phi^{(l_2)}(E_2)\Big)
\end{align*}
Note that for each given $l_1,l_2$ the derivatives of $\tanh$ are uniformly bounded in $\RR$ and can therefore be regarded as universal constants.
Due to the constraints on the degrees of the monomials in the derivatives of $\Phi$ analogous to before, we conclude
\begin{align*}
\|\partial^l_1\hat w\|_{L^1},\|\partial^l_2\hat w\|_{L^1}&=\cO\left(\left(\frac{-\log(\eta)}\delta\right)^l\right)\\
\|\partial^{(l,l)}\hat w\|_{L^1}&=\cO\left(\left(\frac{-\log(\eta)}{\delta}\right)^{2l}\right)\,.
\end{align*}
Note that due to the uniform bound $1$ on $|\hat w|$ and $|\hat g|$ by Equaton~\eqref{eq:hatwWindow} and Lemma~\ref{lem:filter} respectively, we have $\|\hat g\|_{L^1}, \|\hat w\|_{L^1}\le1$

Plugging the above scalings into the asymptotic formulas for the previously derived block-encodings we arrive at the following result:
\begin{theorem}For the function $f$ as described above, the Lindbladian in Definition~\ref{def:jumps}
\[
\cL(\rho)=-i[G,X]+\sum_{a\in\kA}\left(L_a \rho L_a^\dagger-\frac12\{L_a^\dagger L_a,\rho\}\right)
\]
has fixed-point
\[
\sigma=f(H)/\Tr[f(H)]
\]
of $e^{t\cL}$ can be simulated to $\eps$-error in diamond norm in gate complexity
\[
\frac{t|\kA|^2}{\sqrt[k+1]\eps}\polylog(t/\eps)\poly\left(\frac{-\log(\eta)S}{\delta}\right)\,.
\]
times the cost of the gates $\Prep_{g/w}$ and the cost of Hamiltonian simulation for time
\[
\sqrt[k+1]{\frac{|\kA|}\eps}\poly\left(\frac{-\log(\eta)S}{\delta}\right)\,.
\]
\end{theorem}

\paragraph{Relation to trace norm error for window and ground states}
Note that so far, we have obtained perfect KMS-detailed balance for states defined via continuously differentiable functions. With the application of preparation of microcanonical ensembles in mind it is interesting to ask how well the resulting normalized fixed-points approximate the desired output states (defined in terms of the perfect window function $\chi_{[b,c]}$) in 1-norm.
Adversarial cases of Hamiltonians with spectrum concentrated near the edge of the window can be thought of so the error bounds will depend on a more refined understanding of the spectral density.
Note that for the preparation of inherently differentiable function $f$, this discussion is not needed.
This can include alternative definitions of microcanonical ensemble states that only require concentration around the target energy but not necessarily a perfectly flat distribution in a window.

Given $f$, an approximation to the window function as in the previous section, we compute
\begin{align*}
\left\|\frac{\chi_{[b,c]}(H)}{\Tr[\chi_{[b,c]}(H)]}-\frac{f(H)}{\Tr[f(H)]}\right\|_1&=
\frac{1}{\tr[f(H)]}\int_{I\setminus[b,c]}f(E)d\mu_H(E)\\
&\quad+\left(\frac1{\tr[\chi_{[b,c]}(H)]}-\frac1{\tr[f(H)]}\right)\int_{[b,c]}d\mu_H(E)\\
&=\cO\left(\frac1{\tr[\chi_{[b,c]}(H)]}\int_{I\setminus[b,c]}f(E)d\mu_H(E)\right)\\
&=\cO\left(\frac1{\tr[\chi_{[b,c]}(H)]}\left(\eta\int_{I\setminus[b-\delta,c+\delta]}d\mu_H(E)+\int_{[b-\delta,b]\cup[c,c+\delta]}d\mu_H(E)\right)\right)\\
&\le\cO\left(\frac1{\tr[\chi_{[b,c]}(H)]}\left(\eta+\int_{[b-\delta,b]\cup[c,c+\delta]}d\mu_H(E)\right)\right)\,.
\end{align*}
where $d\mu_H(E)$ is the normalized spectral measure and $\tr$ is the normalized trace.
Therefore, in order to obtain good approximations of a sharp window function we require that the spectrum has at least inverse polynomial measure on the target window $[b,c]$ and at most inverse polynomial decaying (in $\delta$) measure in the edges $[b-\delta,b]$ and $[c,c+\delta]$.

In addition, note that for isolated points in the spectrum including \emph{gapped ground states}, where the integral of the spectral measure over a delta interval can be assumed \emph{exactly} zero, even an exponentially small weight in the target interval can be dealt with in polynomial time by choosing an exponentially small $\eta$. In this case $[b,c]$ needs to be chosen such that it contains the isolated point, while no other energy eigenstates lie in $[b-\delta,c+\delta]$.

\section{Conclusion}
In this work, we extend constructions of KMS-detailed-balance Lindbladians to a general class of fixed-point states.
The general framework includes arbitrary states expressible as positive functions of the Hamiltonian up to technical conditions on their differentiability.
As a concrete application, we study window states using appropriate differentiable approximations, capturing both microcanonical ensembles and gapped ground states.

The question of mixing times, while addressed in certain regimes for Gibbs samplers in trivial phases, remains technically challenging and is left open for the constructions introduced in our work.
However, when applied to the Boltzmann function, our proposed algorithm essentially reduces to these prior works.

We also note an alternative approach that we do not pursue in this work. By interpreting the logarithm of the state as a modified parent Hamiltonian of a Gibbs state, one could directly apply existing Gibbs sampling algorithms  \cite{chen2023efficientexactnoncommutativequantum,ding2025polynomialtimepreparationlowtemperaturegibbs}.
Expressing this approach in terms of the original Hamiltonian requires an implementation of the parent Hamiltonian using quantum singular value transformations  (QSVT) \cite{Gilyen_2019}, which is expected to incur comparable costs to our method.
The resulting algorithms would nevertheless differ in structure.

This idea should be distinguished from directly applying singular value transformation to a block encoding of $H$, which gives a block encoding of $f(H)$ for some function $f$. While such approaches can be efficient in an end-to-end sense, the induced approximation properties differ substantially from ours. In particular, our method enables exponential suppression of parts of the spectrum. This feature is not naturally captured in standard QSVT approaches and  QSVT is not expected to yield efficient preparation of thermal or ground states without access to good initial states. 

Window-state preparation is also closely connected to the computational problem of determining the density of states (DOS), the number of eigenvalues of a local Hamiltonian in a given interval. It is known that computing the DOS is equivalent to computing the quantum partition function with respect to computational hardness~\cite{bravyi2022quantum}. However, no such equivalence is known for the corresponding state-preparation question; determining whether the hardness of preparing window-states is equivalent to that of preparing Gibbs states may also be an interesting future direction.

\section*{Acknowledgments}
We thank Zhiyan Ding and Oles Shtanko for insightful discussions.

\clearpage
\printbibliography

\appendix
\section{Tail bounds}
We evaluate some standard decay and truncation bounds for the Fourier coefficients.
\begin{lemma}\label{lem:fourierTailBounds}
For periodic, $k$ times continuously differentiable functions with periodic derivatives\footnote{In particular $f^{(l)}(-S)=f^{(l)}(S)$ $\forall l\le k$} $\hat f:[-S,S]^2\to\RR$, we have the following tail bounds on the Fourier series:
\begin{align*}
    \left|\hat f(E_1,E_2)-\sum_{(n_1,n_2)\in[-M,M]^2}f_{n_1,n_2}e^{i(n_1E_1+n_2E_2)\tau}\right|&\le\left(\frac{8}{\tau^k}\left\|\partial^{(k,k)}\hat f\right\|_{L^1}+\left\|\partial_1^k\hat f\right\|_{L^1}+\left\|\partial_2^k\hat f\right\|_{L^1}\right)\frac{2M^{1-k}}{\tau^k(k-1)}\\
\end{align*}
where $\tau=\pi/S$, $M\ge1$, and we use the convention for the Fourier coefficients
\begin{align*}
f_{n_1,n_2}&=\frac1{4S^2}\int_{I\times I}e^{-i(n_1E_1+n_2E_2)\tau}\hat f(E_1,E_2)d(E_1,E_2)
\end{align*}
with $I=[-S,S]$.
\end{lemma}
\begin{proof}
Via partial integration we obtain the following decay bound on the Fourier coefficients for any $1\le l_1,l_2\le k$
\begin{align*}
|f_{n_1,n_2}|&=\left|\frac1{4S^2}\int_{I\times I}e^{-i(n_1E_1+n_2E_2)\tau}\hat f(E_1,E_2)d(E_1,E_2)\right|\nonumber\\
&\le\frac{1}{n_1^{l_1} n_2^{l_2}}\frac1{4S^2\tau^{l_1+l_2}}\int_{I\times I} \left|\partial^{(l_1,l_2)}\hat f(E_1,E_2)\right|d(E_1,E_2)\nonumber\\
&=\frac{1}{n_1^{l_1} n_2^{l_2}}\frac1{\tau^{l_1+l_2}}\left\|\partial^{(l_1,l_2)}\hat f(E_1,E_2)\right\|_{L_1}
\end{align*}
Here, the boundary terms disappear due to the periodicity assumption on the derivatives and equivalently
\begin{align}\label{eq:coeffiBound1D}
|f_{0,n}|&\le\frac{1}{n^{k}}\frac{1}{4 S^2\tau^{k}}\int_{I\times I}\left|\partial^{k}\hat f(E_1,E_2)\right|d(E_1,E_2)\nonumber\\
&\le\frac{1}{n^k\tau^k}\left\|\partial^k\hat f\right\|_{L_1}
\end{align}
We need the following standard bound on a sum over polynomially decaying coefficients in $\ZZ^2$:
\begin{equation}\label{eq:sumBound1D}
    \sum_{n\ge M+1}\frac1{n^k}\le\int_M^\infty \frac{1}{x^k}dx=\frac{1}{k-1}M^{1-k}
\end{equation}
and
\begin{align}\label{eq:sumBound2D}
\sum_{(n_1,n_2)\in(\ZZ_{\neq0})^2\setminus[-M,M]^2}\frac{1}{n_1^{k_1}n_2^{k_2}}\le2\sum_{n_1=M+1}^\infty\frac1{n_1^{k_1}}2\sum_{n_2=1}^\infty\frac{1}{n_2^{k_2}}+2\sum_{n_1=M+1}^\infty\frac1{n_1^{k_1}}2\sum_{n_2=1}^\infty\frac{1}{n_2^{k_2}}\\
\le 4\frac{1}{k_1-1}M^{1-k_1}\left(1+\frac{1}{k_2-1}\right)+4\frac{1}{k_2-1}M^{1-k_2}\left(1+\frac{1}{k_1-1}\right)
\,.
\end{align}
Putting Equations \eqref{eq:sumBound1D} and \eqref{eq:sumBound2D} together we obtain the following tail bound for a $k$ times continuously differentiable function
\begin{align*}
\sum_{(n_1,n_2)\in\ZZ^2\setminus[-M,M]^2}\left|f_{n_1,n_2}\right|&\le\frac1{\tau^{2k}}\left\|\partial^{(k,k)} \hat f\right\|_{L^1}\frac{16M^{1-k}}{k-1}\\
&\phantom{=}+\left(\left\|\partial_1^k\hat f\right\|_{L^1}+\left\|\partial_2^k\hat f\right\|_{L^1}\right)\frac{2M^{1-k}}{\tau^k(k-1)}\\
&\le\left(\frac{8}{\tau^k}\left\|\partial^{(k,k)}\hat f\right\|_{L^1}+\left\|\partial_1^k\hat f\right\|_{L^1}+\left\|\partial_2^k\hat f\right\|_{L^1}\right)\frac{2M^{1-k}}{\tau^k(k-1)}\,.
\end{align*}
\end{proof}
\begin{lemma}\label{lem:fourierL1Bounds}
Under the same conditions as in Lemma~\ref{lem:fourierTailBounds}, the Fourier coefficients obey the following $l_1$ bounds.
\begin{align*}
\|f_{n_1,n_2}\|_{l^1}&\le10(1+\|\hat f\|_{L^1})\left(1+\frac2{\tau^k}\left(\frac{8}{\tau^k}\left\|\partial^{(k,k)}\hat f\right\|_{L^1}+\left\|\partial_1^k\hat f\right\|_{L^1}+\left\|\partial_2^k\hat f\right\|_{L^1}\right)\right)^{\frac2{k+1}}
\end{align*}
\end{lemma}
\begin{proof}
Note that the Fourier coefficients are bounded by the $\|\cdot\|_{L_1}$-norm of the function so for any $M$, we can combine these trivial bounds for the coefficients up to $|n_1|,|n_2|\le M$ and the previous tails bounds :
\begin{align*}
\|f\|_{l^1}&\le\|\hat f\|_{L^1} (2M+1)^2+\left(\frac{8}{\tau^k}\left\|\partial^{(k,k)}\hat f\right\|_{L^1}+\left\|\partial_1^k\hat f\right\|_{L^1}+\left\|\partial_2^k\hat f\right\|_{L^1}\right)\frac{2M^{1-k}}{\tau^k(k-1)}\\
&\le\|\hat f\|_{L^1} 9M^2+\left(\frac{8}{\tau^k}\left\|\partial^{(k,k)}\hat f\right\|_{L^1}+\left\|\partial_1^k\hat f\right\|_{L^1}+\left\|\partial_2^k\hat f\right\|_{L^1}\right)\frac{2M^{1-k}}{\tau^k(k-1)}
\end{align*}
and we set 
\[
M=\left\lceil\sqrt[k+1]{\frac{\frac{2}{\tau^k}\left(\frac{8}{\tau^k}\left\|\partial^{(k,k)}\hat f\right\|_{L^1}+\left\|\partial_1^k\hat f\right\|_{L^1}+\left\|\partial_2^k\hat f\right\|_{L^1}\right)}{\|\hat f\|_{L^1}}}\right\rceil.
\]
Then
\[
\|f\|_{l^1}\le10(1+\|\hat f\|_{L^1})\left(1+\frac2{\tau^k}\left(\frac{8}{\tau^k}\left\|\partial^{(k,k)}\hat f\right\|_{L^1}+\left\|\partial_1^k\hat f\right\|_{L^1}+\left\|\partial_2^k\hat f\right\|_{L^1}\right)\right)^{\frac2{k+1}}.
\]
\end{proof}

\section{Filter function}
We discuss the filter function chosen in our approach.
Our choice is
\begin{align*}
\nu_{C,\zeta}(x)&=\exp(-\sqrt{1+C(1-\cos(x\pi/\zeta))})\\
\kappa(E_k,E_l)&=\nu_{C,\zeta}(E_k-E_l)\,.
\end{align*}

We summarize the required properties in the following lemma.
\begin{lemma}\label{lem:filter}
For the choice $C=S^2L^2/32$, and $\zeta=S$, we have the following
\begin{itemize}
    \item $\kappa$ is $2S$-periodic in both arguments
    \item $\hat g(E_k,E_l)\le1$
    \item $\nu_{C,\zeta}(x)$ is analytic with 
    \[
\left\|\frac1{\nu_{C,\zeta}(x)}\frac{d^k}{dx^k}\nu_{C,\zeta}(x)\right\|_{L^\infty}= \mathcal O\left(\left(\frac{C\pi}{\zeta}\right)^k\right)
\]
\end{itemize}
\end{lemma}
\begin{proof}
For the first point note that for $|x|\le\zeta$
\[
\exp(-\sqrt{1+C(1-\cos(x\pi/\zeta))})\le \exp(-\sqrt{C(1-\cos(x\pi/\zeta))})\le\exp(-\frac{|x|\sqrt{2C}}{\zeta})
\]
since $1-\cos(x\pi)=2\sin^2(x\pi/2)\ge 2x^2$.

So we get for $|E_k-E_l|\le S$
\begin{align*}
\hat g(E_k,E_l)&=\exp((\Phi(E_k)-\Phi(E_l))/4)\nu_{C,S}(E_k-E_l)\\
&\le\exp(L|E_k-E_l|/4)\exp(-\frac{|E_k-E_l|\sqrt{2C}}{S})\le1
\end{align*}
by choosing $C=L^2S^2/32$.
For $2S\ge E_k-E_l>S$ we use the periodicity and $E_k-E_l-2S\in[-S,S]$ to bound
\begin{align*}
\hat g(E_k,E_l)&=\exp((\Phi(E_k-2S)-\Phi(E_l))/4)\nu_{C,S}(E_k-E_l-2S)\\
&\le \exp(L|E_k-E_l-2S|/4)\exp(-\frac{|E_k-E_l-2S|\sqrt{2C}}{S})\le 1
\end{align*}
and analogously for $-2S\le E_k-E_l\le S$.

For the differentiability, note that $\nu$ is analytic on $\RR$ with all derivatives uniformly bounded due to the periodicity. We merely consider the scaling of the derivatives with the parameters,
\[
\left\|\frac1{\nu_{C,\zeta}(x)}\frac{d^k}{dx^k}\nu_{C,\zeta}(x)\right\|_{L^\infty}= \mathcal O\left(\left(\frac{C\pi}{\zeta}\right)^k\right)\,.
\]
\end{proof}
\end{document}